\documentclass{article}

\usepackage[nonatbib, final]{nips_free}

\usepackage[final]{nips_free}

\usepackage[utf8]{inputenc} % allow utf-8 input
\usepackage[T1]{fontenc}    % use 8-bit T1 fonts
\usepackage{hyperref}       % hyperlinks
\usepackage{url}            % simple URL typesetting
\usepackage{booktabs}       % professional-quality tables
\usepackage{amsfonts}       % blackboard math symbols
\usepackage{nicefrac}       % compact symbols for 1/2, etc.
\usepackage{microtype}      % microtypography

\usepackage{graphicx}
\usepackage{epsfig}
\usepackage{amssymb}
\usepackage{amsmath}
\usepackage{amsthm}
\usepackage{subfigure}

\usepackage[backend=bibtex,style=ieee,doi=false, isbn=false, url=false, eprint=false, maxnames=3]{biblatex}
\newtheorem{theorem}{Theorem}
\newtheorem{lemma}{Lemma}

\newcommand*\samethanks[1][\value{footnote}]{\footnotemark[#1]}

\usepackage{dsfont}

\newcommand{\Ind}{ \mathds{1} }

\newcommand{\G}{ \mathcal{G}}
\newcommand\distrib{ {\sim} }	% "Distributed as".
\newcommand{\N}{\mathbb{N}}
\newcommand\lintersect{ {\cap} }

\newcommand\lunion{ {\cup}}

\newcommand{\bs}[1]{\boldsymbol{#1}}

\DeclareMathAlphabet\mathbfcal{OMS}{cmsy}{b}{n}

\newcommand{\Graph}[1]{\mathit{\mathbf{\bs{#1}}}}
\newcommand{\Set}[1]{\mathit{#1}}

\newcommand{\pCoDO}{p_{CoDO}}

\newcommand{\E}{\mathbb{E}}

\newcommand{\Hypergeometric}{\mathrm{Hypergeometric}}

\newcommand{\var}[1]{#1}
\newcommand{\rvar}[1]{\uppercase{#1}}

\begin{document}

\title{Combining Density and Overlap (CoDO): A New Method for Assessing the Significance of Overlap Among Subgraphs}

% The \author macro works with any number of authors. There are two
% commands used to separate the names and addresses of multiple
% authors: \And and \AND.
%
% Using \And between authors leaves it to LaTeX to determine where to
% break the lines. Using \AND forces a line break at that point. So,
% if LaTeX puts 3 of 4 authors names on the first line, and the last
% on the second line, try using \AND instead of \And before the third
% author name.

\author{
	Abram Magner\thanks{Corresponding authors: anmagner@illinois.edu}   \thanks{Authors contributed equally.} \\
    Coordinated Science Lab, UIUC \\
    Champaign, IL 61820 \\
    \texttt{anmagner@illinois.edu}
    % Affiliation
    % Address
    % \texttt{email}
    \And
    Shahin Mohammadi \samethanks \\
    Dept. of Computer Science, Purdue University \\
    West Lafayette, IN 47907 \\
    \texttt{mohammadi@purdue.edu}
    \And
    Ananth Grama \\
    Dept. of Computer Science, Purdue University \\
    West Lafayette, IN 47907 \\
    \texttt{ayg@cs.purdue.edu}
}

\maketitle
\begin{abstract}
Algorithms for detecting clusters (including overlapping clusters) in graphs have received significant attention in the research community. A closely related important aspect of the problem -- quantification of statistical significance of overlap of clusters, remains relatively unexplored. This paper presents the first theoretical and practical results on quantifying statistically significant interactions between clusters in networks. Such problems commonly arise in diverse applications, ranging from
social network analysis to systems biology. The paper addresses the problem of quantifying the statistical significance of the observed overlap of the two clusters in an Erd\H{o}s-R\'enyi graph model. The analytical framework presented in the paper assigns a $p$-value to overlapping subgraphs by combining information about both the sizes of the subgraphs and their edge densities in comparison to the corresponding values for their overlapping component. This $p$-value is demonstrated to have excellent discrimination properties in real applications and is shown to be robust across broad parameter ranges.

Our results are comprehensively validated on synthetic, social, and biological networks. We show that our framework: (i) derives insight from both the density and the size of overlap among communities (circles/pathways), (ii) consistently outperforms state-of-the-art methods over all tested datasets, and (iii) when compared to other measures, has much broader application scope. In the context of social networks, we identify highly interdependent (social) circles, and show that our predictions are highly co-enriched with known user features. In networks of biomolecular interactions, we show that our method identifies novel cross-talk between pathways, sheds light on their mechanisms of interaction, and provides new opportinities for investigations of biomolecular interactions.
\end{abstract}

%%% A category with the (minimum) three required fields
%%\category{H.4}{Information Systems Applications}{Miscellaneous}
%%%A category including the fourth, optional field follows...
%%\category{D.2.8}{Software Engineering}{Metrics}[complexity measures, %performance measures]
%%a
%\begin{CCSXML}
%<ccs2012>
%<concept>
%<concept_id>10003752.10010061.10010069</concept_id>
%<concept_desc>Theory of computation~Random network models</concept_desc>
%<concept_significance>500</concept_significance>
%</concept>
%<concept>
%<concept_id>10010405.10010444.10010087.10010091</concept_id>
%<concept_desc>Applied computing~Biological networks</concept_desc>
%<concept_significance>500</concept_significance>
%</concept>
%% <concept>
%% <concept_id>10002950.10003648.10003662.10003666</concept_id>
%% <concept_desc>Mathematics of computing~Hypothesis testing and confidence interval computation</concept_desc>
%% <concept_significance>300</concept_significance>
%% </concept>
%</ccs2012>
%\end{CCSXML}
%
%\ccsdesc[500]{Theory of computation~Random network models}
%\ccsdesc[500]{Applied computing~Biological networks}
%% \ccsdesc[300]{Mathematics of computing~Hypothesis testing and confidence interval computation}
%
%\printccsdesc
%
%%\terms{Theory}
%\keywords{  } 	

\section{Introduction}

Quantification of statistical significance of an observed artifact in data is a
critical aspect of data analytics applications, particularly at scale. The
statistical significance of an artifact is often quantified in terms of its
$p$-value -- the probability that a random event caused the observed artifact.
Clearly, the lower this probability, the higher the statistical significance.
Statistical significance, in terms of its $p$-value is defined with respect to
a prior, which models the background (random) distribution. Selecting a prior
that is distant from the data renders all observations significant (i.e., they
have very low $p$ values). Conversely, selecting a prior that is identical to
the data renders all observations insignificant (i.e., their $p$ values approach
1). A suitable combination of a statistical prior along with an appropriate
$p$ value formulation provides discrimination of significant artifacts from
those that are not.

$p$-values of observed artifacts may be estimated using analytic or
sampling techniques. Sampling techniques for artifacts in graphs often
require large numbers of samples for convergence. Analytic techniques,
on the other hand, pose significant challenges for derivation of tight
bounds, particularly for complex graph models. Analytic techniques have
been successfully demonstrated in the context of characterizing the
significance of dense subgraphs~\cite{koyuturk2007}, distributions of subgraph
diameters, and frequencies of network motifs such as triangles. Sampling based
techniques have been demonstrated in conjunction with non-parametric graph
priors to characterize the alignment of networks~\cite{scvshs}.

In this paper, we focus on the important and challenging problem of
characterizing the statistical significance, in terms of its $p$-value,
of the observed overlap between two dense subgraphs. This problem arises
in a number of applications -- for instance: (i) in understanding when
an overlap between two communities in a social network suggests a
social tie; (ii) in assessing whether overlap among two groups of
like-minded movie viewers suggests a shared interest; and (iii) in
determining whether crosstalk between two sets of interacting
biomolecules corresponds to a functional property.

The problem of analytically characterizing a suitable $p$-value for the overlap
of two subgraphs, even for simple graph models such as Erd\H{o}s-R\'enyi, is a
hard one. One can use simpler abstractions to characterize observed
overlap using Hypergeometric Tails (HGT) or by characterizing the significance
of the overlapping subgraph independently (i.e., ignoring the existence of
the two sub-graphs). However, these simpler abstractions lack the
discriminating power necessary for most applications. The Hypergeometric Tail does not consider network structure, therefore its applicability is limited. Characterizing the significance of the overlapping subgraph does not consider cluster sizes, relying, instead on high density within the overlapping subgraph. Our method, on the other hand, allows for cases in which densities of each component subgraph and their overlap are not significant by themselves with respect to the ambient graph, but the edges inside the two subgraphs are unusually concentrated inside the overlap.  This aspect of our method gives it excellent robustness across wide parameter ranges, as well as diverse applications. We present a detailed theoretical foundation for quantifying the $p$-value
of the overlap of two subgraphs in an Erd\H{o}s-R\'enyi graph. Our formulation
considers the sizes of the two sub-graphs, their densities, as well as
the size and density of the overlap set. In doing so, it provides excellent
discrimination across the parameter space. 

We comprehensively validate the suitability of our selected prior and our framework for quantifying statistical significance of overlaps on synthetic as well as real networks. Using an Erd\H{o}s-R\'enyi random graph with implanted clusters of different size, density, and overlap, we experimentally evaluate the behavior of our framework with respect to different parameters. Having established the baseline behavior, we apply our method to the ego networks for Facebook, Google+, and Twitter networks~\cite{NIPS2012}. We identify social circles with significant overlap and show that these circles are highly enriched in common features. Finally, using a well curated network of		
interacting human proteins and a given set of functional pathways, we
identify statistically significant crosstalk between these pathways. We
show that the crosstalk identified correlates very well with known
biological insights, while also providing a number of novel observations for
future investigation. We also demonstrate that our measure significantly
outperforms other characterizations of statistical significance.

\section{Background and Related Work}
% General problem.
%% Old approaches.
Evaluating the significance of overlap among a pair of given subgraphs is 
a challenging problem that can be addressed from different viewpoints, 
depending on how significance is quantified. The most common method for assessing
overlap is to ignore the underlying interactions and focus on the size 
of the overlap (number of vertices), compared to the sizes of 
subgraph pairs. In this approach, one often uses Fisher's exact test, 
also known as the hypergeometric test \cite{handbookbiostats}. This 
method is described in more detail in Section~\ref{HGTSection}. 
Another way to define significance of overlap is in terms of density. 
Here, an overlap is defined to be significant if there are ``many'' 
edges in the overlap set, compared to a null model. In this class of 
methods the size of overlap set is assumed to be fixed and density of 
overlap is assessed, conditioned on the fixed size. The key aspect of this
measure is the the density of overlap. 

There has been extensive research focused on the algorithmic task of 
identifying dense components in graphs (see, e.g., \cite{densesurvey2010} for a survey).  
Regarding the study of dense subgraphs in random graphs, the papers by Arratia in 1990 
and Koyut\"urk et al. in 2007 \cite{arratia1990,koyuturk2007} are perhaps the most relevant.  Both study
the typical sizes of subgraphs of a given density
in Erd\H{o}s-R\'enyi graphs, motivated by the study of protein interaction 
networks.  The former also gives a result on overlaps between dense subgraphs,
but only in the context of dense graphs, unlike social or biological networks that are 
sparse and the number of edges and nodes are typically of the same order.  The 
distributions of the maximum sizes of cliques and independent sets in random graphs 
have also been studied, essentially completely \cite{bollobas}.  The typical behavior 
of the overlaps between independent sets in the case of sparse
random graphs (equivalently, cliques in extremely dense random graphs) 
has recently been studied \cite{sudan2014}. However, the range of parameters for
which the results are derived differs significantly from ours, rendering these results
less relevant to our applications of interest. While these contributions are fundamental to our 
understanding of statistics regarding the distribution of dense components in a graph, 
{\em none} of them directly address the problem of assessing the overlap among subgraphs. 
Here, we combine the two approaches based on sizes of the overlapping subgraphs and the 
density of the overlap compared to the constituent subgraphs to analytically derive a  
\emph{p}-value for the overlap that considers both its size and density. 

From an applications point of view, this problem has important implications in social networks,
systems biology, financial transactions, and network security. In systems 
biology, as pathways associated with specific biological functions are fully resolved, 
there is increasing need to assess the extent of overlap/ cross-talk among these pathways. 
Pathway interactions have been shown to play a key role in development and progression 
of cancers \cite{Sun2015, Wang}. Different groups have focused on identifying the 
cross-talk map among pathways \cite{SGA, Tegge2015}. However, to the best of our 
knowledge, none of these methods directly use the overlap subgraph to infer 
pathway interactions. Another important application of our method in life sciences
is in geneset enrichment analysis. While the hypergeometric test remains the most 
widely used method for identifying functional enrichment of a set of differentially 
expressed genes, most recent methods do not focus only on the geneset overlap but 
also try to incorporate network context to evaluate functional importance of 
genesets \cite{Mitrea2013, Dong2016}. However, a majority of these methods rely on 
computing empirical \emph{p}-values instead of providing a closed-form solution. 
Our method is the first method that quantifies statistical significance within a network
context, and in doing so, provides significantly higher discriminating power than prior
methods, while using significantly fewer computing resources.

%%%%%%%%%%%%%%%%%%%%%%%%%%%%%%%%%%%%%%%%%%%%%%%%%%%%%%%
\section{Statistical significance of subgraph overlaps}

Given a pair of overlapping dense subgraphs in a graph, our goal is to derive
a measure of the extent to which we should be surprised by their overlap.  
The appropriate framework for this is that of \emph{statistical 
significance}.  In the most general setting, we have a random variable $\rvar{X}$ 
on a probability space and an observed value $\hat{X}$ for $X$, and we want a 
measure of the surprise in observing $\hat{X}$.  To the observed value $\hat{X}$, 
we associate a \emph{p-value}, which is given by the probability that $\rvar{X}$ 
takes a value \emph{as extreme or more extreme} than $\hat{X}$ (the notion of 
extremity depends on the range of the random variable).  If the $p$-value is less 
than some fixed threshold, the observed value is said to be \emph{statistically 
significant} with respect to the distribution of $\rvar{X}$ (the \emph{prior}); that is, the 
observation $\hat{X}$ is unlikely to have occurred by chance.  

For a given random variable (which need not even take values in a partially ordered 
set), there may be many non-equivalent notions of $p$-value that take into account 
various pieces of observable information about $\rvar{X}$.  In general, a notion of 
statistical significance is preferable to another if it has more 
\emph{discriminating power}, in the sense that the $p$-value decays smoothly as 
observations become more extreme (this allows for statistically significant 
observations to be detected \emph{and} to be ranked by relative significance).

In what follows, we will describe three formulations of statistical significance 
for subgraph overlaps.  The first two, which are already present in literature, 
take into account only partial information: the first takes into account only the 
sizes of the subgraphs and their intersection; the second takes into account only 
the size and density of the intersection of the overlapping subgraphs.  The third, 
which we call \emph{CoDO} and which is the main topic of this paper, combines subgraph
sizes and densities of individual subgraphs along with their their overlaps, to yield a test with significantly
greater discriminating power. 

The basic setup in all three tests is as follows: we have a
graph $\Graph{G} \distrib \G(n, p)$ (where $\G(n, p)$ denotes the 
Erd\H{o}s-R\'enyi distribution on graphs of size $n$, with edge 
probability $p$) and subgraphs $A$ and $B$ with $X = A \setminus B$, $Y 
= B\setminus A$, and $Z = A\lintersect B$.  We also have the following
definitions and notation relating to density:
For any subgraph $S \subseteq \Graph{G}$, we denote by $E(S)$ the
set of edges in $S$, and $e(S) = |E(S)|$.  For a pair of subgraphs 
$S_1, S_2$, we denote the set of edges between nodes in $S_1$ and
nodes in $S_2$ by $E(S_1, S_2)$, and $e(S_1, S_2) = |E(S_1, S_2)|$.
We denote by $\delta(S)$ the \emph{density} of $S$:
\begin{align*}
	\delta(S) = \frac{e(S)}{ {{|S|}\choose 2}}.
\end{align*}
Similarly, for $S_1$, $S_2$,
\begin{align*}
	\delta(S_1, S_2) = \frac{e(S_1, S_2)}{|S_1||S_2|}.
\end{align*}

We now describe several $p$-value formulations.
\subsection{Hypergeometric Tail (HGT)}
\label{HGTSection}
The \emph{hypergeometric tail} $p$-value takes into account only the
\emph{size} of the overlap $Z$ of $A$ and $B$.  It is defined by considering the 
following probabilistic experiment (note that this hypothetical experiment is introduced
for the purpose of defining a particular probability distribution as the outcome of some
random process; it is \emph{not} meant to be implemented by a user of our methods): we fix a subset
$\hat{B} \subseteq V(G)$ of size $|\hat{B}| = |B|$, and we draw, uniformly at 
random from $V(\Graph{G})$, a subset $\hat{A}$ of nodes of size $|\hat{A}| = |A|$.  
The size $|\hat{A} \lintersect \hat{B}|$ of the overlap is then hypergeometrically 
distributed, so the hypergeometric
$p$-value is given by the probability that the overlap in this experiment has size 
at least that of the observed overlap $Z$.  That is,
\begin{align*}
	\Pr[|\hat{A} \lintersect \hat{B}| \geq |Z|]
    = 1 - F(|Z|-1, n, |\hat{B}|, |\hat{A}|),
\end{align*}
where $F(x, y, z, w)$, for arbitrary $x \in \N \lunion \{0\}$, denotes the probability that a hypergeometrically distributed 
random variable with population size $y$, number of successes $z$, and number of 
trials $w$, takes a value at least $x$.

%%%%%%%%%%%%%%%%%%%
\subsection{Erd\H{o}s-R\'{e}nyi Density Model (ERD)}
A second approach to scoring of $p$-values of overlaps takes into account only the 
size and density of the overlap set $\Set{Z}$.  It is defined as the probability 
that there exists in $\Graph{G}$ a subgraph of density $\delta(\Set{Z})$ and size 
at least $|\Set{Z}|$.  This is given by
\begin{align*}
	\Pr[\exists \Set{H} \subseteq \Graph{G}, |\Set{H}|\geq |\Set{Z}| ~:~ \delta(\Set{H}) = \delta(\Set{Z})].
\end{align*}
Upper bounds on this probability have been worked out by, e.g., Arratia in 1990 and Koyut\"urk et al. in 2007 
\cite{arratia1990, koyuturk2007}.  The main tool for this is the \emph{first moment method} 
\cite{alonspencer}: the event that there exists a subgraph with a given property is 
precisely equal to the event that the number of subgraphs with the given property 
is at least $1$.  The first moment method consists of an application of Markov's 
inequality to conclude that
\begin{align*}
	\Pr[\#\{\Set{H} ~:~ |\Set{H}| \geq |\Set{Z}|, \delta(\Set{H}) =  \delta(\Set{Z}) \} \geq 1] 
    \leq 
    \E[\#\{\Set{H} ~:~ |\Set{H}| \geq |\Set{Z}|, \delta(\Set{H}) = \delta(\Set{Z}) \}]
\end{align*}
The expected value is then easily calculated using linearity of expectation. 
For a fixed value of $\delta(Z)$, the threshold value of $|Z|$ below which the 
expected value tends to $\infty$ as $n\to\infty$ and above which it tends to $0$ 
turns out to be $\frac{2\log(n)}{\kappa(\delta(Z), p)}$, where $\kappa(a, b)$ is 
given by
\begin{align}
	\kappa(a, b) = a\log(a/b) + (1-a)\log( (1-a)/(1-b)),
   	\label{KappaDefinition}
\end{align}
the relative entropy between Bernoulli random variables with parameter $a$ and $b$, 
respectively.
The first moment method shows that overlaps with density $\delta(Z)$ and size at 
least $(1+\epsilon)|Z|$, for any fixed positive $\epsilon$, are unlikely to occur 
randomly.  The \emph{second moment method}, which takes into account the dependence 
between events defined on overlapping subgraphs, can be used to give a matching 
probabilistic lower bound, which shows that, with high probability (i.e., with 
probability asymptotically tending to $1$), the maximum size of any $\delta(Z)$-
dense subgraph is asymptotically equivalent to $\frac{2\log(n)}{\kappa(\delta(Z), 
p)}$.

Because of the very sharp transition, as $|Z|$ increases, from insignificant to 
significant, this formulation has the undesirable property that it has little 
discriminating power.  In particular, the first moment upper bound implies that
\begin{align*}
	\Pr[\#\{\Set{H} ~:~ |\Set{H}| \geq |\Set{Z}|, \delta(\Set{H}) =  \delta(\Set{Z}) \} \geq 1]
    \leq e^{-\Theta(|\Set{Z}|^2)}.
\end{align*}
That is, the $p$-value decays superexponentially as $|Z|$ increases.

%%%%%%%%%%%%%%%%%%%
\subsection{Combining Density/Overlap (CoDO)}
\label{sec:CoDO_theory}

The measures of statistical significance presented above are not well suited
to real-world applications, where the densities and sizes of the the dense
components may vary significantly, and where we desire smoother transitions.
Motivated by these shortcomings, we propose an alternate formulation that
combines information about the size \emph{and} density of the 
overlapping set, defined by the following probabilistic 
experiment:
\begin{itemize}
	\item
    	Consider a set $\Graph{V}$ of vertices, with size $\var{n}$,
        and a distinguished subset $\Set{\hat{B}}$ of size 
        $|\Set{B}|$.  Choose uniformly at random from $\Graph{V}$ a
        subset $\Set{\hat{A}}$ of size $|\Set{A}|$.
	\item
    	In the set $\Set{\hat{A}} \lunion \Set{\hat{B}}$, choose uniformly
        at random $|e(\Set{A}\lunion \Set{B})|$ edges to insert.
\end{itemize}
Then, we define the combined density/overlap $p$-value to be the probability
that the resulting overlap $\Set{\hat{A}} \lintersect \Set{\hat{B}}$ has size at 
least that of the observed overlap $|\Set{Z}|$ \emph{and} that the density of the 
overlapping set $\delta(\Set{\hat{A}} \lintersect \Set{\hat{B}})$ is at least the 
observed overlap density $\delta(\Set{Z})$.
That is, it is given by
\begin{align*}
	\pCoDO
    = \Pr[|\Set{\hat{A}} \lintersect \Set{\hat{B}}| \geq |\Set{Z}| \lintersect \delta(\Set{\hat{A}} ~\lintersect~ \Set{\hat{B}}) \geq \delta(\Set{Z})]
\end{align*}
By conditioning on the size of the overlap, we can get an explicit formula for this 
$p$-value in terms of hypergeometric tails:
\begin{align}
	\label{pCoDOSumFormula}
	\pCoDO 
    = \sum_{j=|\Set{Z}|}^{\min\{|\hat{A}|, |\hat{B}|\}}
      \Pr\left[|\hat{A}\lintersect \hat{B}| = j\right]
      \cdot \Pr\left[\delta(\hat{A}\lintersect \hat{B}) \geq \delta(\Set{Z}) \big|  |\hat{A} \lintersect \hat{B}| = j \right]
\end{align}
The first factor in the terms of the sum is easily seen to be a hypergeometric 
probability mass:
\begin{align*}
	|\hat{A}\lintersect\hat{B}|
    \distrib \Hypergeometric(n, |B|, |A|),
\end{align*}
which is the number of successes when one draws without replacement
$|A|$ samples from a collection of $n$ items, $|B|$ of which are successes.  
Analogously, the second factor
is a hypergeometric probability tail, since
\begin{align*}
	e(\hat{A}\lintersect \hat{B})
    \distrib \Hypergeometric\left( {{ |A| + |B| - j}\choose 2}, {j\choose 2}, e(A\lunion B)\right).
\end{align*}
Thus, the sum is explicitly computable, and tail bounds can be applied if approximations are desired.
Note that the expression (\ref{pCoDOSumFormula}) involves the observed sizes and densities of the constituent subgraphs,
as well as the overlaps, through the distributions of $|\hat{A} \lintersect \hat{B}|$ and $\delta(\hat{A} \lintersect \hat{B})$.

%%%%%%%%%%%%%%%%%%%%%%%%%%%
\subsubsection{Characteristics of the CoDO $p$-Value}

We investigate in more detail some interesting properties of the CoDO
$p$-value defined above.

An elementary observation is that, whenever the density of overlap subgraph $\delta(Z)$
is low enough, the conditional probabilities in the $\pCoDO$ formula
degenerate to $1$, and it becomes equivalent to the hypergeometric
tail $p$-value. Similarly, when the overlap size $|Z|$ is small, all probabilities degenerate to $1$.

% Behavior for fixed overlap densities.
Next, we look at the behavior of $p_{CoDO}$ when we fix $\delta(Z) =
\rho_Z \in (0, 1)$ and vary the observed overlap size $|Z|$.  We
consider $z = \frac{|Z|}{M} \in (0, 1)$, where we define 
$M = \min\{|A|, |B|\}$.

% Monotone decreasing.
We have the following lemma:
\begin{lemma}[Monotone decrease as $|Z|$ increases]
	\label{MonotonicityLemma1}
	For $\rho_Z$ fixed as above, for $z_1 < z_2$, we have
    \begin{align*}
    	\pCoDO(z_1, \rho_Z) 
        \geq \pCoDO(z_2, \rho_Z). 
    \end{align*}
\end{lemma}
\begin{proof}
	We have
    \begin{align*}
    	\pCoDO(z_1, \rho_Z) - \pCoDO(z_2, \rho_Z) 
        =&~ \sum_{j=z_1M}^{M} 
          \Pr[|\hat{A}\lintersect\hat{B}| = j] \Pr\left[\delta(\hat{A}\lintersect \hat{B}) \geq \rho_Z \big| |\hat{A}\lintersect\hat{B}| = j \right] \\
          &-\sum_{j=z_2M}^{M} 
          \Pr[|\hat{A}\lintersect\hat{B}| = j] \Pr\left[\delta(\hat{A}\lintersect \hat{B}) \geq \rho_Z \big| |\hat{A}\lintersect\hat{B}| = j \right] \\
    	&=  \sum_{j=z_1M}^{z_2M - 1} 
          \Pr[|\hat{A}\lintersect\hat{B}| = j] \Pr\left[\delta(\hat{A}\lintersect \hat{B}) \geq \rho_Z \big| |\hat{A}\lintersect\hat{B}| = j \right]
    \end{align*}
    Since all of the terms of the final sum are at least $0$,
    this completes the proof.
\end{proof}

It then becomes interesting to investigate the point at which $\pCoDO$
transitions from insignificant to significant as a function of $z$. 
To do this we, define the \emph{threshold point} $\theta(\rho_Z)$ by
\begin{align*}
	\theta(\rho_Z) = z_* 
    = \inf \{ z \in (0, 1) ~:~ \pCoDO(z, \rho_Z) \leq 1/2 \}.
\end{align*}
Note that, whenever it exists, the threshold point is unique, because of
the monotonicity property just proven.  In principle, the asymptotic 
value of the threshold point, as a function of $\rho_Z$, is explicitly 
computable. However, the fact that $\E\left[e(\hat{A}\lintersect \hat{B}) \big| 
|\hat{A}\lintersect\hat{B}=j|\right]$ is an increasing function of $j$ 
whenever $j$ is large enough makes this somewhat subtle.  
The following upper bound can be established:
\begin{theorem}
	\label{ThresholdTheorem1}
	With $\rho_Z \in (0, 1)$ fixed and 
    \[
    	M = \min\{|A|, |B|\} \gg \sqrt{n},
    \]    
    we have, for any fixed $\delta > 0$,
    \begin{align*}
    	\theta(\rho_Z) 
        \leq \frac{\E[|\hat{A}\lintersect\hat{B}|]}{M}(1 + \delta).
    \end{align*}
\end{theorem}
To prove this, we will need the following lemma giving large deviation 
bounds for the hypergeometric distribution \cite{chvatal1979}.  It is a consequence
of the Hoeffding inequality \cite{dembozeitouni}.
\begin{lemma}[Hypergeometric large deviations]
	\label{HypergeometricLargeDeviationsLemma}
	Let $X \distrib \Hypergeometric(N, K, n)$ (the parameters denote
    the population size, number of successes, and number of trials, 
    respectively).  Let $p = K/N$.  Then we have the following tail
    bound:
    \begin{align*}
    	\Pr[X \geq (1 + \delta)\E[X]]
        \leq \exp( -n\kappa((1+\delta)p, p) ),
    \end{align*}
    where $\delta$ is any number greater than $0$,
    and $\kappa(a, b)$ is given by (\ref{KappaDefinition}).
\end{lemma}

\begin{proof}[of Theorem~\ref{ThresholdTheorem1}]
	Let $z = (1+\delta)\frac{\E[|\hat{A}\lintersect\hat{B}|]}{M}$.
	We have the following upper bound on $\pCoDO(z, \rho_Z)$.
    \begin{align}
    	\pCoDO(z, \rho_Z)
        \leq \sum_{j=zM}^{M} \Pr[|\hat{A}\lintersect\hat{B}|=j]p_* 
        = \Pr[|\hat{A}\lintersect\hat{B}| \geq zM]p_*,
        \label{IntermediateExpr}
    \end{align}
    where we define
    \begin{align*}
    	p_* = \max_j\left\{\Pr\left[  \delta(\hat{A}\lintersect \hat{B}) \geq  \rho_Z \big| |\hat{A}\lintersect\hat{B}|=j \right]\right\},
    \end{align*}
    where the maximum is taken over all terms $j$ in the sum
    (this is allowed because all of the probabilities are 
    non-negative).
    
    Next, we apply Lemma~\ref{HypergeometricLargeDeviationsLemma}
    to further upper bound (\ref{IntermediateExpr}) by 
    \begin{align*}
    	\pCoDO(z, \rho_Z)
        \leq
        \Pr[|\hat{A}\lintersect \hat{B}| \geq 
        	(1+\delta)\E[|\hat{A}\lintersect \hat{B}| ]]p_* 
        \leq 
        \exp(-|A|\kappa((1+\delta)|B|/n, |B|/n )) p_*.
    \end{align*}
    We now examine the role of the relative entropy factor in the
    exponent.  As $|B|/n \xrightarrow{n\to\infty} 0$, its first term
    asymptotically dominates:
    \begin{align*}
    	\kappa( (1+\delta)|B|/n, |B|/n)
        = (1+\delta)|B|/n \log(1+\delta) + o(1).
    \end{align*}
    Using the fact that $|B| \geq M \gg (\sqrt{n})$, we then have
    that the entire exponent tends to $-\infty$.  That is, provided
    that $n$ is large enough, $\pCoDO(z, \rho_Z)$ is 
    arbitrarily small and, in particular, is less than $1/2$.
\end{proof}
Note that, in the proof of this theorem, we did not use any information
about the conditional probability in the $\pCoDO$ formula.  This is 
another incarnation of the phenomenon that the conditional expectation
of the number of edges in $\hat{A}\lintersect \hat{B}$ increases with
$j$.  Thus, broadening the applicability of this bound requires a more
careful study of how the conditional probabilities behave as $j$
increases.

% Behavior for fixed overlap size.
The behavior of $\pCoDO(z, \rho_Z)$ for fixed and varying $\rho_Z$ is
also of interest.  We have an analogue of Lemma~\ref{MonotonicityLemma1}:
\begin{lemma}[Monotone decrease as $\rho_Z$ increases]
	\label{MonotonicityLemma2}
    For fixed $z \in (0, 1)$, we have, for $\rho_1 < \rho_2$,
    \begin{align*}
    	\pCoDO(z, \rho_1)
        \geq \pCoDO(z, \rho_2).
    \end{align*}
\end{lemma}
\begin{proof}
	We have
    \begin{align*}
    	&\pCoDO(z, \rho_1)
        - \pCoDO(z, \rho_2)  \\
        &= \sum_{j=zM}^{M}
        	\Pr[|\hat{A}\lintersect \hat{B}|=j]\cdot 
            \left[
            	\Pr\left[\delta(\hat{A}\lintersect\hat{B}) \geq \rho_1
                	\big| |\hat{A}\lintersect\hat{B}|=j\right] \right.
                \left.~~- \Pr\left[\delta(\hat{A}\lintersect\hat{B}) \geq \rho_2
                	\big| |\hat{A}\lintersect\hat{B}|=j \right]    
            \right]
    \end{align*}
    By monotonicity of the CDF of a random variable, the difference 
    of conditional probabilities is non-negative, so that the entire
    sum is positive.  This completes the proof.
\end{proof}
As an analogue of the definitions for fixed $\rho_Z$ and varying $z$,
we can define a threshold point $\phi(z)$. 
We can also give analogous bounds
for the location of the threshold point.  The details of
the derivation are entirely similar to those for fixed $z$ and varying $\rho_Z$.

\section{Experimental Validation}

We provide a detailed experimental evaluation of CoDO, in comparison with other measures, in the contest of synthetic datasets, datasets derived from social networks, and datasets from biomolecular interactions. We use synthetic datasets to characterize the dependence of our results of choices of parameters. We use the other datasets to demonstrate the robustness and application scope of our framework, even in cases where the networks do not follow an Erdos-Renyi model.

\subsection{Assessing Overlap in Synthetic Graphs}
\label{sec:syn_results}

To evaluate the performance of each measure (and associated method) on a controlled dataset,
we create a synthetic test case that models the behavior of functional modules at a smaller
scale. We sample an ER graph $\Graph{G}_R$ from $\G(n, p)$ with $n=80$ and
$p=\frac{3}{n}$. We embed two modules with the same high density ($\rho_{\Set{A}} =
\rho_{\Set{A}} = 10p$) of size 50 and 40 vertices in $\Graph{G}_R$ with varying levels
of overlaps and density. For the overlap, the expected value of hypergeometric distribution
for this setting is $\frac{n_\Set{A}n_\Set{A}}{n} = 25$, where $n_A$ and $n_B$ are the
number of vertices in modules $\Graph{A}$ and $\Graph{B}$, respectively. We chose overlap
sizes of 20 and 30 vertices, which are below and above the expected value of number of overlapping
vertices, respectively, to represent low and high overlaps. For simulating the density
of overlap region, we tested a sparse overlap by setting
$\rho_{\Set{Z}} = 2p = \frac{\rho_{\Set{A}}}{5}$ and a dense overlap by setting
$\rho_{\Set{Z}} = 10p = 2\rho_{\Set{A}}$. For each of these four settings, we computed
\textit{HGT, ERD}, and \textit{CoDO} \emph{p}-values, the results of which are illustrated
in Figure~\ref{fig:results_synthetic}. Each subfigure represents the adjacency matrix
of the constructed test case, with modules $\Graph{A}$ and $\Graph{B}$ color-coded as
blue and green, respectively. The overlap set of the two modules is separately marked
in red. The first notable observation in our experiments is that \textit{ERD} has a
sharp transition from \emph{p}-value of 1 to 0 as we vary $\rho_{\Set{Z}}$ from
$\rho$ to $2\rho$. On the other hand, \textit{CoDO} exhibits a smooth transition as
we increase $\rho_{\Set{Z}}$.

\begin{figure}[!h]
\centering
\hfil%
\subfigure[Insignificant Overlap/Insignificant Density\label{fig:inSigOver-inSigDen}]{\includegraphics[width=.5\columnwidth]{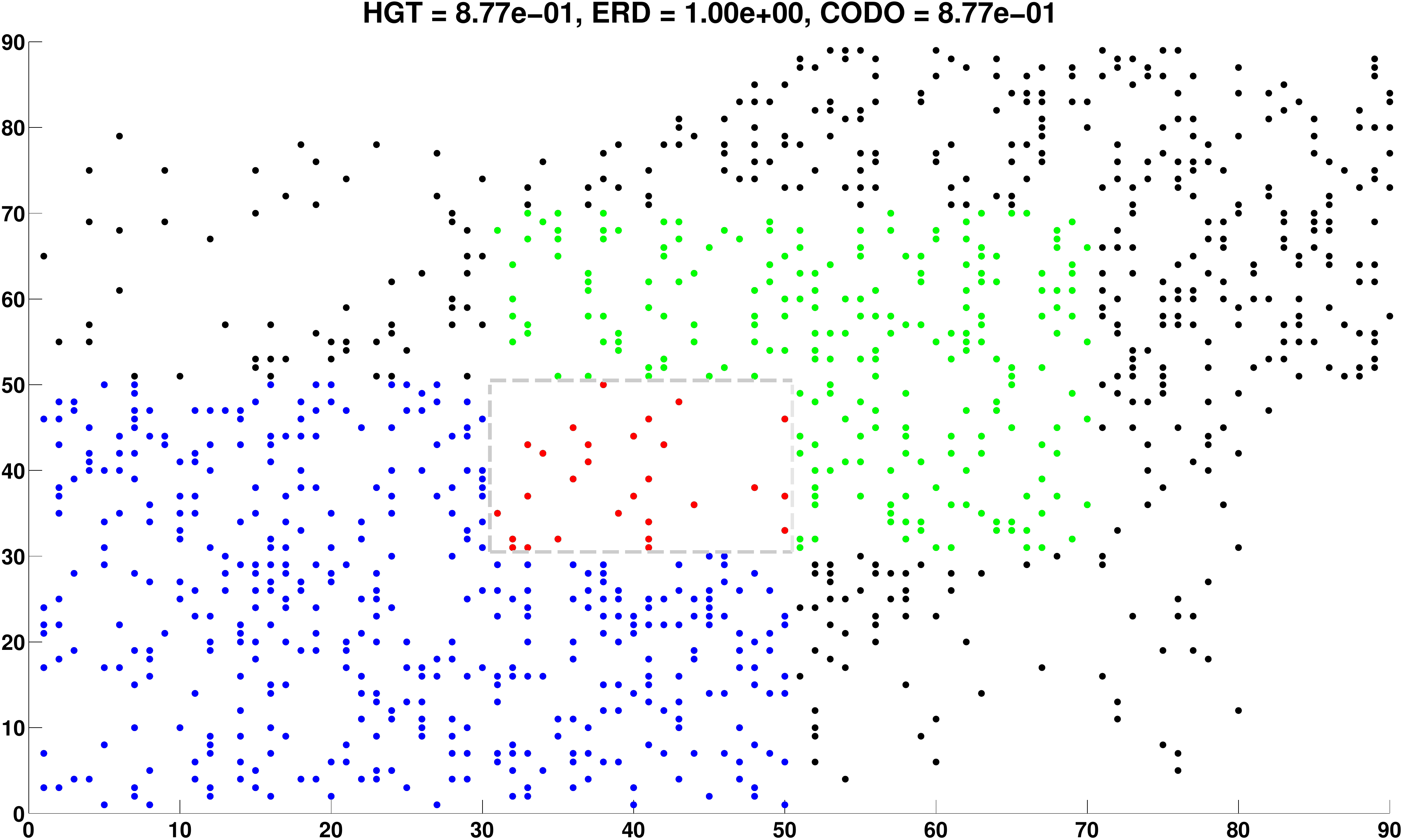}}%
\hfil%
\subfigure[Insignificant Overlap/Significant Density\label{fig:inSigOver-SigDen}]{\includegraphics[width=.5\columnwidth]{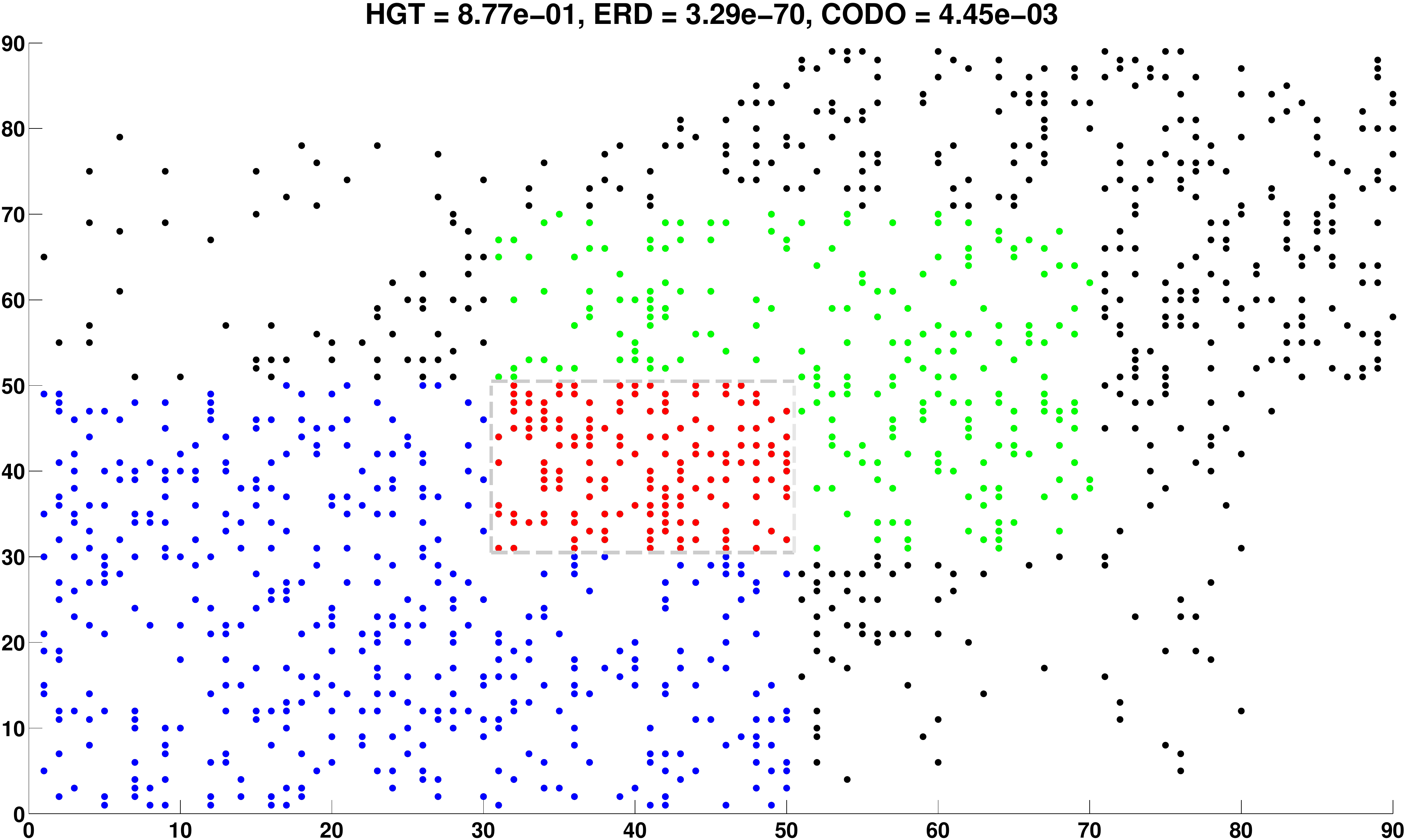}}%
\hfil%
\subfigure[Significant Overlap/Insignificant Density\label{fig:SigOver-inSigDen}]{\includegraphics[width=.5\columnwidth]{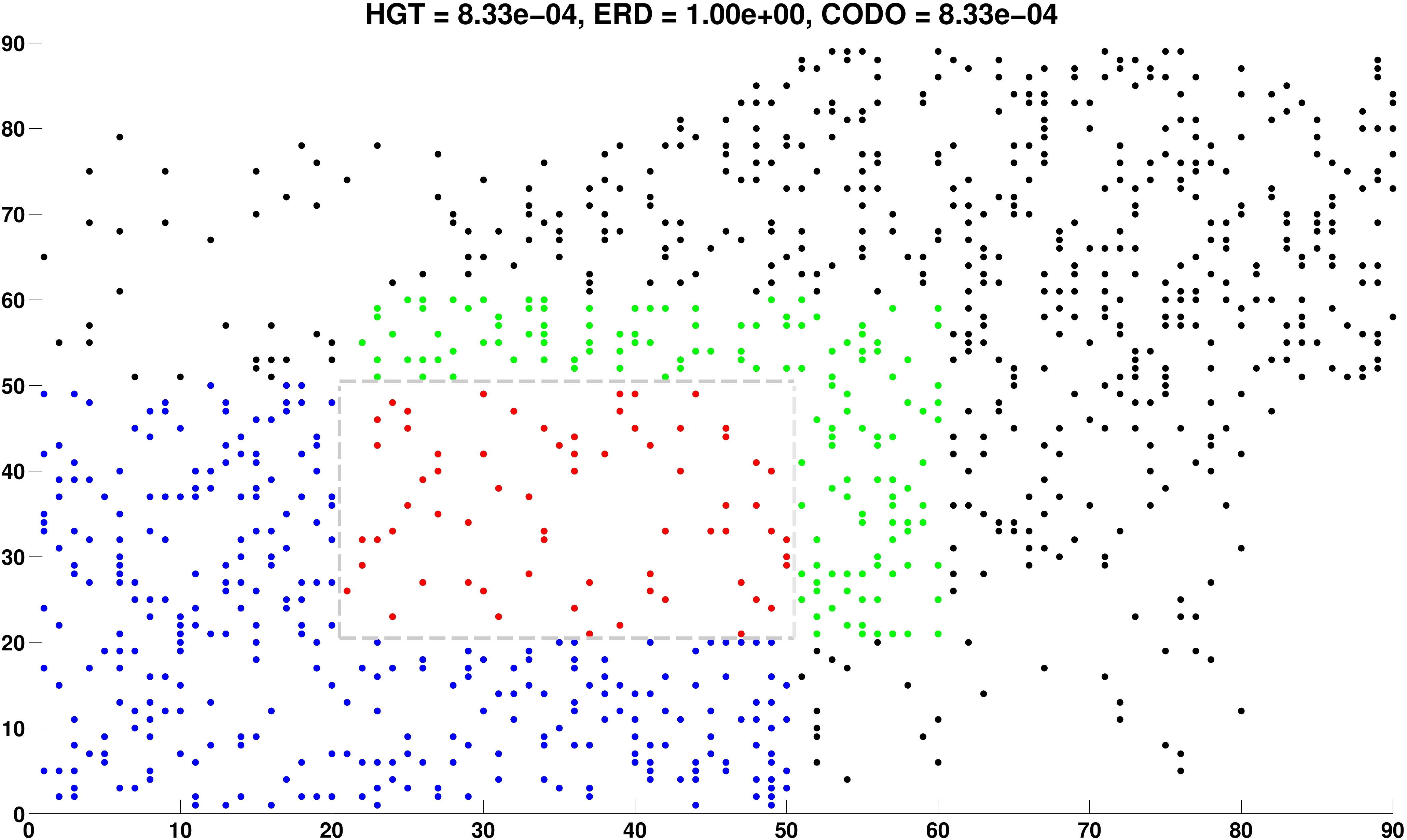}}%
\hfil%
\subfigure[Significant Overlap/Significant Density\label{fig:SigOver-SigDen}]{\includegraphics[width=.5\columnwidth]{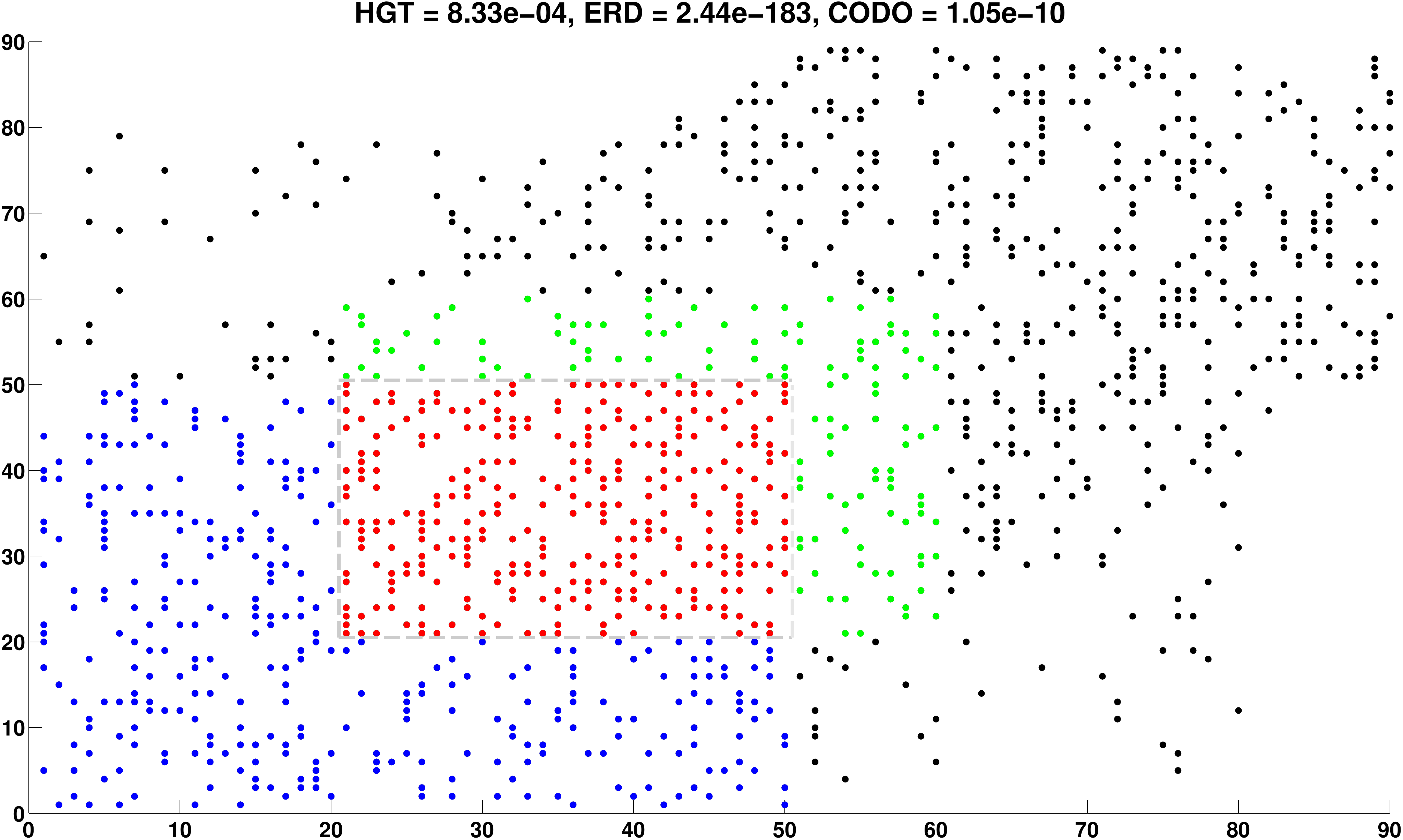}}%
\caption{Synthetic graphs with varying levels of density/overlap. Illustrated here is the adjacency matrix of generated graphs. Black, green, and blue dots (edges) belong to the background, the first, and the second implanted subgraphs, respectively, whereas red dots represent edges in the overlap.}
\label{fig:results_synthetic}
\end{figure}

When neither the overlap, nor its density are significant, as in Figure~\ref{fig:inSigOver-inSigDen},
the \emph{p}-values from all methods are close to 1 (as expected). When the overlap size
is significant but density of overlap is not (Figure~\ref{fig:SigOver-inSigDen}), \textit{HGT}
and \textit{CoDO} yield identical significant results, while \textit{ERD} does not identify the
overlap as significant (\emph{p}-value approaching 1). On the other hand, when size of overlap
is not significant but its density is significant, as in Figure~\ref{fig:inSigOver-SigDen},
\textit{HGT} \emph{p}-value is close to one, whereas both \textit{ERD} and \textit{CoDO}
identify the overlap as significant, the difference being that \textit{ERD} \emph{p}-value
abruptly changes from $1$ to $3.29 \times 10^{-70}$, but \textit{CoDO} smoothly transitions
from $8.77 \times 10^{-1}$ to $4.45 \times 10^{-3}$ demonstrating much better discriminating
power. Finally, when both the size of the overlap and its density are significant,
\textit{CoDO} \emph{p}-value is more significant than in both cases in
Figures~\ref{fig:inSigOver-SigDen} and Figure~\ref{fig:SigOver-inSigDen}, and it is more
significant that \textit{HGT} \emph{p}-value alone.

\subsection{Significance of Overlap Among Social Circles}
Users in social networks are connected not only to their close friends, but also family members, schoolmates, and colleagues, among others. Organizing friends into communities, or \textit{social circles}, is one of the common approaches to organizing, predicting, and recommending contacts. A user's circles are typically \textit{overlapping}, and a user can belong to one, many, or none of the circles. Given a set of overlapping circles, we are interested in assessing whether there is a significant interdependency among the circles, or that the observed overlap is simply due to the existence of \textit{bridge nodes} or \textit{party alters}.

To evaluate our method, we use three sets of social networks derived from Facebook, Google+, and Twitter\cite{NIPS2012}. Each of these networks is centered around a focal user, or \textit{ego}, along with all other users it directly connects to. This collection of neighbor nodes is also known as the set of \textit{alters}. The induced subgraph of friendships among \textit{alter nodes}, is known as an \textit{ego net}. The Facebook dataset is \textit{fully observed}, in the sense that each user was asked to manually identify all coherent groups among their friends. We refer to these coherent groups of friends as \textit{circles}. In contrast to the Facebook dataset, circles in Google+ and Twitter are restricted to publicly visible and explicitly designated circles for each ego net.

To establish a ground truth for interdependency among pairs of circles, we use common features of users to assign a \emph{p}-value to their feature overlap. For a given ego net $G=(V, E)$ together with a $K$-dimensional feature vector for each node, we compute the total number of feature pairs between any pair of alters as $\pi = K{|V| \choose 2}$. Moreover, we can compute the total number of features common to any two alters as:

\begin{equation}
\xi = \sum_{\substack{i, j \in \Set{V}\\ i < j}} \sum_{f \in \mathcal{F}} \Ind_{f}(i, j)
\end{equation}

where $\Set{V}$ is the set of alter vertices, $\mathcal{F}$ is the set of features, and $\Ind_{f}(i, j)$ is an indicator function that is one, if alter $i$ and alter $j$ share feature $f$, and zero, otherwise. Given an overlap set of $O$, we can similarly define $\pi_O = K{|O| \choose 2}$ and:

\begin{equation}
\xi_O = \sum_{\substack{i, j \in \Set{O}\\ i < j}} \sum_{f \in \mathcal{F}} \Ind_{f}(i, j)
\end{equation} 

Using this notion, we can define the significance of feature overlap among two circles as:

\begin{equation}
p-val(O) = \sum_{x = \xi_O}^{\textbf{min}(\xi, \pi_O)} \frac{{\xi \choose x}{\pi-\xi \choose \pi_O - x}}{{\pi \choose \pi_O}}
\end{equation}

which is simply the tail of the Hypergeometric distribution, measuring the probability of observing $\xi_O$ or more common features in a random set of size $O$ (with $\pi_O$ feature pairs), if we were sampling at random without replacement from all possible feature pairs. For each dataset, we compute the significance of all circle pairs for every ego net and correct for multiple hypothesis testing using \textit{Bonferroni} method. Finally, we define a pair of circles as \textit{significantly correlated} if \emph{p}-value of their overlap is smaller than or equal to the threshold $0.05$. Table~\ref{table:ego_stats} summarizes the statistics of these networks. Here, we remove all ego nets with less than two circles or those without at least one common feature between alters. Among the three datasets, Facebook had the lowest percentage of significant overlaps ($< 25\%$) compared to the other two datasets ($35\%$ and $37\%$). This may be attributed to method used for identifying circles in the Facebook dataset.

\begin{table}[t]
  \caption{Statistics of ego nets}
  \label{table:ego_stats}
  \centering
  \begin{tabular}{lccccc}
    \toprule
	Dataset	&\# Nets	&Avg \# Users	&Avg \# Circles	&\# Over. Circles	&\# Signif. Overlaps\\
    \midrule
	Facebook	&10	&416.7	&19.3	&415	&92\\
	Google+	&124	&1,945.7	&3.5	&1,170	&437\\
	Twitter	&834	&137.6	&4.2	&14,610	&5,055\\
    \bottomrule
  \end{tabular}
\end{table}

Next, to evaluate \textit{HGT, ERD}, and \textit{CoDO} methods, we assess all pairs of overlapping circles and sort them based on their significance. For pairs with similar \emph{p}-value, we randomly order them. Using the gold-standard set of \textit{significantly correlated} circles, computed using feature vector overlaps, we compute the \textit{receiver operating characteristic (ROC)} for each dataset individually. These are presented in Figure~\ref{fig:results_ego}. Each method is annotated with its area under the curve (AUC) to simplify comparison. In all three datasets, \textit{CoDO} outperforms the other two methods. It is worth noting that \textit{CoDO} is designed to capture signals from both network density and the overlap size. To illustrate this point, we emphasize the behavior of \textit{CoDO} over the Facebook dataset as an example. For small FPR values, density has higher signal than overlap (observed as the superior performance of ERD compared to HGT). In this regime, \textit{CoDO} captures this signal and \textit{outperforms} of \textit{ERD}. When \textit{ERD} signal plateaus ($0.1 < FPR$), \textit{CoDO} leverages network structure to outperform (\textit{HGT}).

\begin{figure}[!h]
\centering
\hfil%
\subfigure[Facebook\label{fig:FB}]{\includegraphics[width=.45\columnwidth]{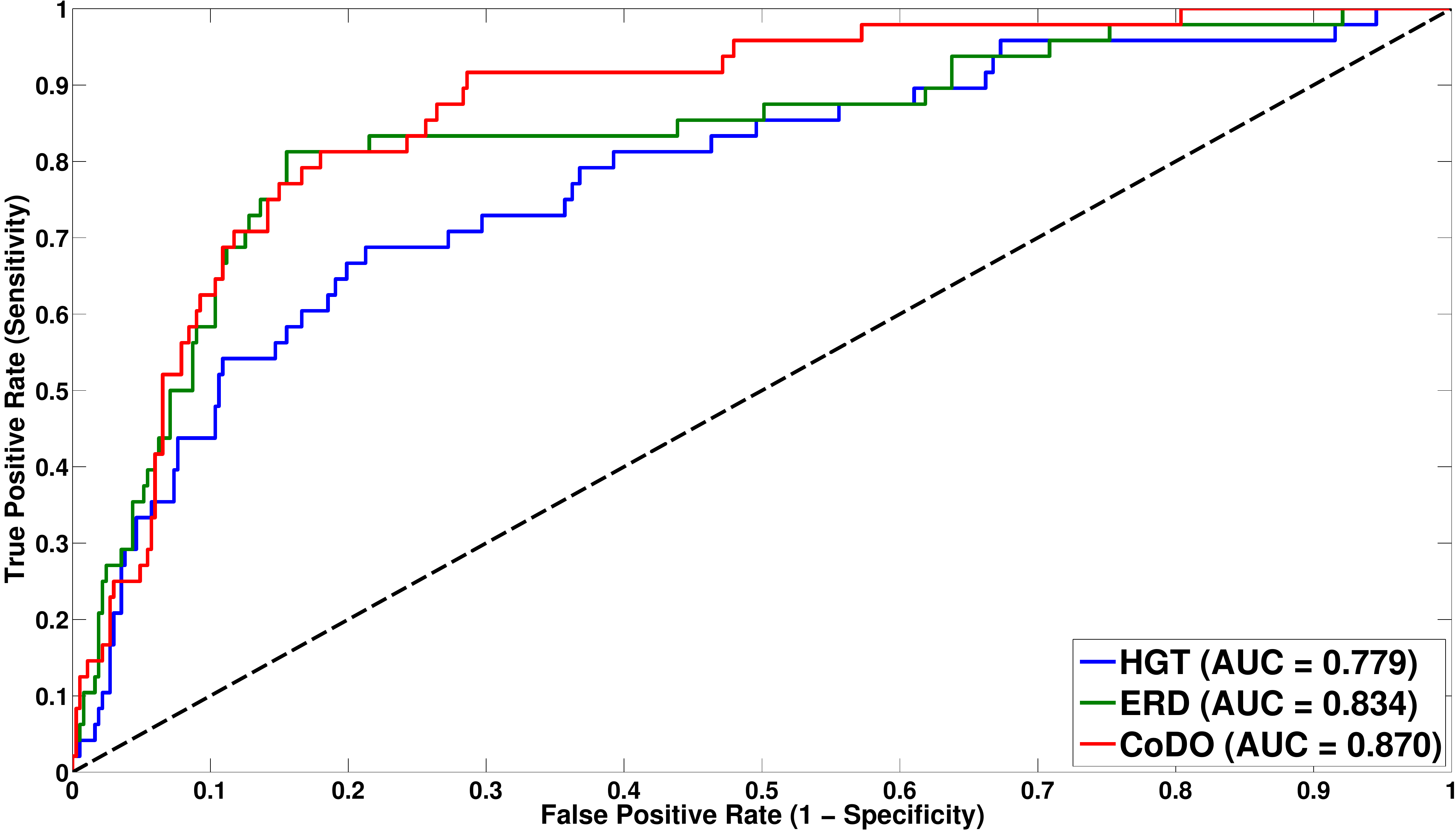}}%
\hfil%
\subfigure[Google+\label{fig:GPlus}]{\includegraphics[width=.45\columnwidth]{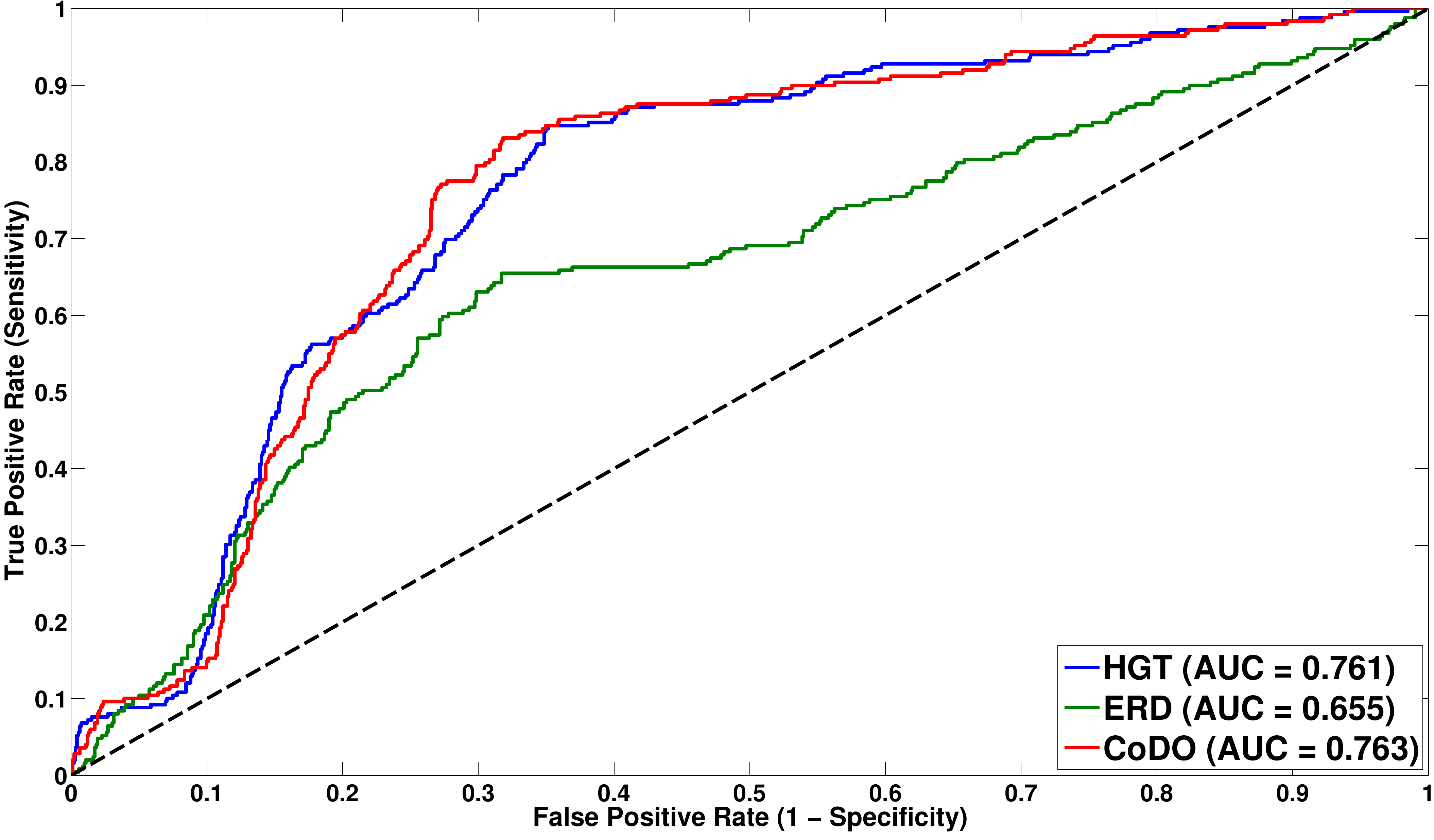}}%
\hfil%
\subfigure[Twitter\label{fig:Twitter}]{\includegraphics[width=.45\columnwidth]{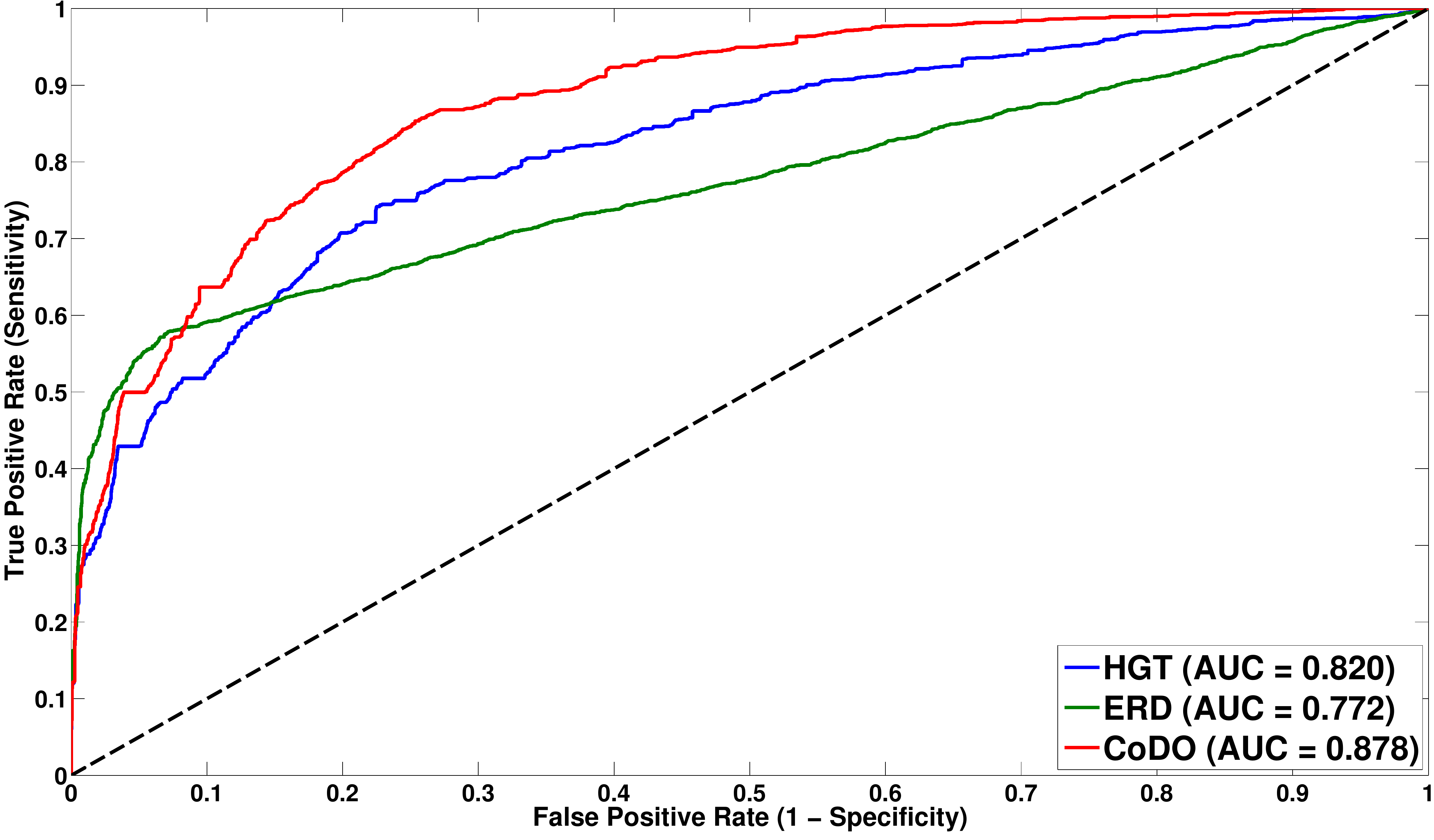}}%
\caption{Overlapping circles in social networks predict common features}
\label{fig:results_ego}
\end{figure}

\subsection{Crosstalk Among Biological Pathways}
\label{sec:KEGG_overlapPvals}

Biological pathways are higher order constructs that perform key cellular functions. These pathways are known
to work in concert and their cross-talk regulates the systems level behavior of
cells. To evaluate the interaction between different pathways and their
shared mechanisms, we seek to construct a pathway-pathway interaction map that represents
the extent of overlap among different pathways. 

We construct a comprehensive human interactome from a recently published dataset by Han \textit{et al.}\cite{ESEA}, resulting in a network of 164,826 interactions among 8,872 proteins. We downloaded the set of KEGG pathways from MSigDB\cite{MSigDB} and filtered pathways that have less than 10 corresponding vertices in the human interactome. The final dataset consists of 186 pathways with an average of $\sim 66$ vertices per pathway. To evaluate the computed \emph{p}-values, we use \textit{co-transcriptional activity of pathway pairs} as a proxy for their functional relatedness. We downloaded the and processed tissue-specific RNASeq dataset from the Genotype-Tissue Expression (GTEx) project \cite{GTEx}, which contains 2,916 samples from 30 different tissues/ cell types.  We summarize tissue-specific expression of each pathway by averaging the expression of all its member genes in each tissue. Then, we define co-transcriptional activity of pathways by computing the Pearson's correlation of these tissue-specific pathway expression signatures.

To evaluate different methods, we first compute the nonparametric Spearman's correlation
between the co-transcriptional activity scores and computed \emph{p}-values in each
method. This yields $0.28, 0.16,$ and $0.41$ correlation scores for \textit{HGT, ERD,}
and \textit{CoDO} methods, respectively, which shows that \textit{CoDO} significantly outperforms
the other two methods with respect to the co-transcriptional activity of significant pathway pairs.
Next, to put the computed \emph{p}-values by \textit{CoDO} in context, we construct a pathway-pathway overlap network by thresholding pairwise overlaps at a stringent cutoff (\emph{p}-value
$\leq 10^{-30}$) and use all pairs with significant \emph{p}-values as edges in the network.
This results in a network of 129 nodes (representing pathways) and 877 interactions
(inferred as the significance of the pathway overlaps) among them. This network is shown
in Figure~\ref{fig:results_pathway_overlap}. We use Cytoscape \cite{Cytoscape} to visualize
the graph and MCODE \cite{MCODE} to cluster the network. Each cluster is color-coded
independently, and four major clusters are manually annotated with the dominant function
in each group. We observe that the set of top-ranked pathway pairs cluster together to form coherent groups with highly coherent functions.
These significant overlaps reveal interesting functional connections, which can be used to understand pathology and identify novel drug targets.

\begin{figure}
\includegraphics[width=0.9\columnwidth]{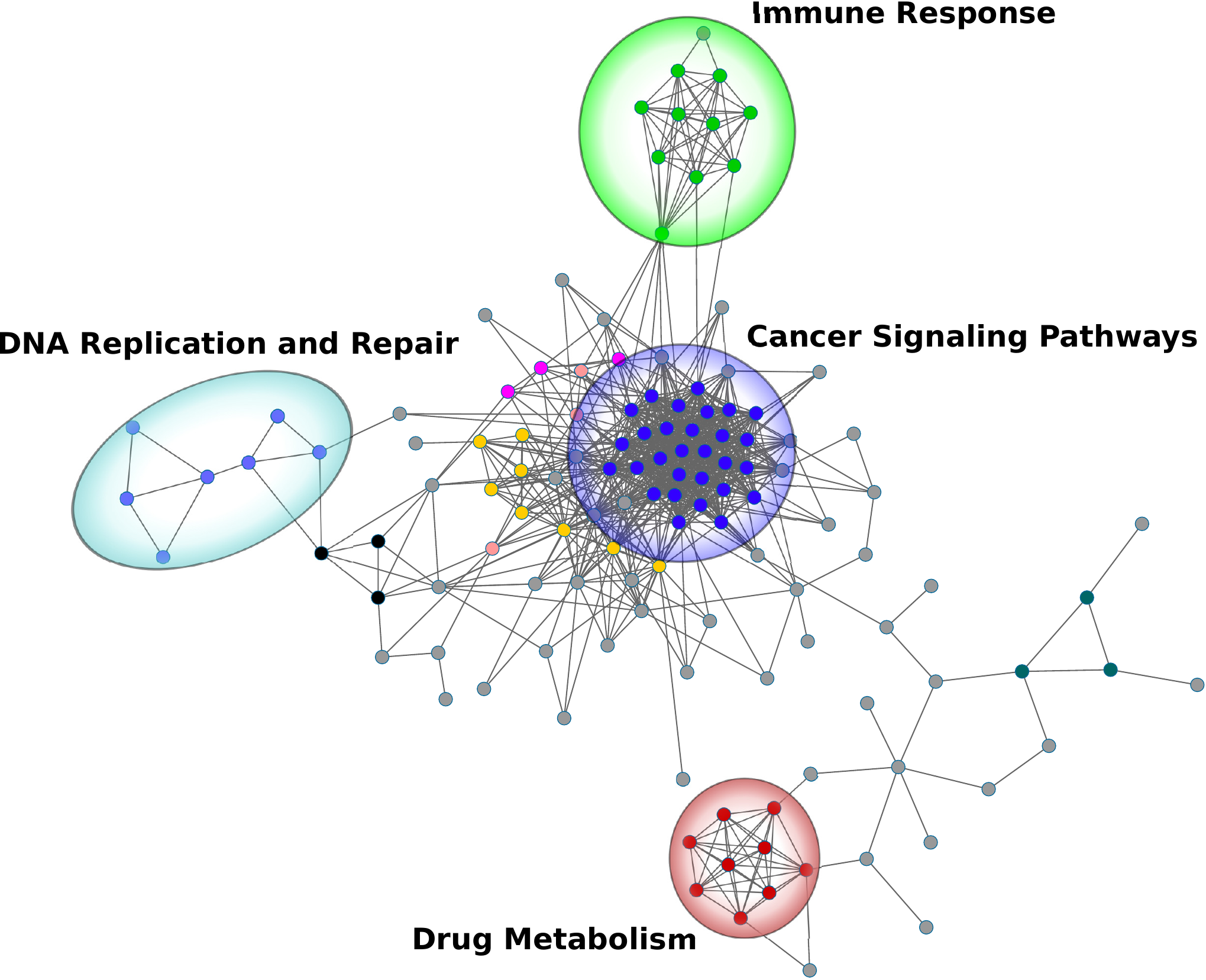}
\caption{Pathway-pathway overlap graph}
\label{fig:results_pathway_overlap}
\end{figure}

\section{Conclusions}

Assessing the statistical significance of observed overlaps between clusters in networks
is an important substrate of many graph analytics problems. In this paper, we present detailed
analysis and derivation of a $p$-value based measure of significance of cluster overlap. 
Unlike previous measures, our measure accounts for cluster sizes, densities, and density
of overlap -- all critical parameters of the problem. We show that our measure provides
excellent discrimination, smooth transition, and robustness to wide parameter ranges.
Using both synthetic and real datasets, we validate excellent performance of our analytical
formulation.

Our work opens a number of avenues for continued explorations. These include formulations
for alternate graph models, algorithms aimed at explicitly maximizing statistical significance
of overlaps, and application validation in other contexts.

%\section{Acknowledgments}
%This work is supported by the Center for Science of Information (CSoI), an NSF Science and 
%Technology Center, under grant agreement CCF-0939370, and by NSF Grant BIO 1124962.

\clearpage
\printbibliography

% \bibliographystyle{abbrv} % {abbrv}
% \bibliography{refs}  % sigproc.bib is the name of the Bibliography in this case

%\printbibliography 

% %%%%%%%%%%%%%%%%%%%%%%%%%%%%%%%%%%%%%%%%%%%%%%%%%%%%%%%%%%
% \appendix
% %Appendix A
% \balancecolumns
% % That's all folks!
\end{document}